\documentclass[conference,a4paper]{IEEEtran}

\usepackage[dvips]{graphicx}
\usepackage{epsfig}
\usepackage[cmex10]{amsmath}
\usepackage{amssymb}
\usepackage{amsthm}
\usepackage{amsfonts}
\usepackage{bm}
\usepackage{cite}
\usepackage[svgnames]{xcolor} % load before pstricks
\usepackage{pstricks,pst-node,pst-plot,pstricks-add}
\usepackage[tight,footnotesize]{subfigure}
\usepackage{algpseudocode}
\usepackage{algorithm}
\usepackage{graphics}
\usepackage{mathrsfs}
\usepackage{psfrag}
\usepackage{tikz}
\usepackage{accents}
\usepackage{bbm}
\graphicspath{{figs/}}

\input{parham.def}

\IEEEoverridecommandlockouts

\begin{document}

\title{Error Exponents of Mismatched\\ Likelihood Ratio Testing}

\author{%
  \IEEEauthorblockN{Parham Boroumand}
  \IEEEauthorblockA{University of Cambridge\\
                    \tt pb702@cam.ac.uk}
  \and
  \IEEEauthorblockN{Albert ~Guill\'en i F\`abregas}
  \IEEEauthorblockA{ICREA \& Universitat Pompeu Fabra\\
                    University of Cambridge\\
                    \tt guillen@ieee.org}
                    \thanks{
This work was supported in part by the European Research Council under 
Grant 725411, and by the Spanish Ministry of Economy and Competitiveness 
under Grant TEC2016-78434-C3-1-R.
}
}

\maketitle

\begin{abstract} 

We study the problem of mismatched likelihood ratio test. We analyze the type-\RNum{1} and \RNum{2} error exponents when the actual distributions generating the observation are different from the distributions used in the test. We derive the worst-case error exponents when the actual distributions generating the data are within a relative entropy ball of the test distributions. In addition, we study the sensitivity of the test for small relative entropy balls.
\end{abstract}

\section{Introduction and Prelimenaries}

Consider the binary hypothesis testing problem \cite{Lehmann} where an observation $\xv=(x_1,\dotsc,x_n)$ is generated from two possible distributions $P^n_1$ and $P^n_2$ defined on the probability simplex $\Pc(\Xc^n)$. We assume that $P^n_1$ and $P^n_2$ are product distributions, i.e., $P_1^n(\xv)=\prod_{i=1}^n P_1(x_i)$, and similarly for $P_2^n$. For simplicity, we assume that both $P_1(x)>0$ and $P_2(x)>0$ for each $x\in\Xc$. 

Let  $\phi: \mathcal{X}^n \rightarrow \{1,2\}$ be a hypothesis test that decides which distribution generated the observation $\xv$. We consider deterministic tests $\phi$ that decide in favor of $P_1^n$ if $\xv\in \Ac_1$, where $\Ac_1\subset \Xc^n$ is the decision region for the first hypothesis. We define $\Ac_2=\Xc^n \setminus \Ac_1$ to be the decision region for the second hypothesis. The test performance is measured by the two possible pairwise error probabilities. The type-\RNum{1} and type-\RNum{2} error probabilities are defined as
\begin{equation}\label{eq:e1}
\epsilon_1 (\phi)= \sum_{\bx \in \Ac_2} P_1^n(\bx),~~~~ \epsilon_2 (\phi)= \sum_{\bx \in \Ac_1} P_2^n(\bx).  
\end{equation}
A hypothesis test is said to be optimal whenever it achieves the optimal error probability tradeoff given by
\begin{equation}\label{eq:trade}
\alpha_\beta = \min_{\phi: \epsilon_2 (\phi) \leq \beta} \epsilon_1 (\phi) .
\end{equation}
The likelihood ratio test defined as
\begin{equation}
\phi_\gamma(\boldsymbol{x})= \mathbbm{1} \bigg \{ \frac{P_2^n(\bx)}{P_1^n(\bx)}  \geq e^{n\gamma} \bigg\}+1.
\end{equation}
was shown in  \cite{Neyman} to attain the optimal tradeoff \eqref{eq:trade} for every $\gamma$. The type of a sequence $\bx= (x_1,\ldots,x_n)$ is $\Th(a)=\frac{N(a|\bx)}{n}$, where $N(a|\bx)$ is the number of occurrences of the symbol $a\in\Xc$ in the string. The likelihood ratio test can also be expressed as a function of the type of the observation $\Th$ as \cite{Cover}
\begin{align}\label{eq:LRTtype}
\phi_{\gamma}(\Th)= \mathbbm{1} \big\{ D(\Th\|P_1)-D(\Th\|P_2)    \geq  \gamma \big\}+1.
\end{align}
where $D(P\|Q)= \sum_{\mathcal{X}} P(x) \log \frac{P(x)}{Q(x)}$ is the relative entropy between distributions $P$ and $Q$.

In this paper, we are interested in the asymptotic exponential decay of the pairwise error probabilities. Therefore, it is sufficient to consider deterministic tests The optimal error exponent tradeoff $(E_1,E_2)$ is defined as
\begin{align}\label{eq:tradefix}
E_2(E_1) \triangleq \sup \big\{E_2\in \mathbb{R}_{+}: \exists \phi , \exists n_0 \in \mathbb{Z}_+  \  \text{s.t.} \  \forall   n>n_0 \nonumber \\
  \epsilon_1(\phi) \leq e^{-nE_1}  \quad \text{and} \quad  \epsilon_2(\phi) \leq e^{-nE_2}\big\}.
\end{align}
By using the Sanov's Theorem \cite{Cover,Dembo}, the optimal error exponent tradeoff $(E_1,E_2)$, attained by the likelihood ratio test, can be shown to be \cite{Blahut,Hoeffding}
\begin{align}
E_1(\phi_{\gamma})=\min_{Q \in \mathcal{Q}_1(\gamma)} D(Q\|P_1)\label{eq:min1},\\
E_2(\phi_{\gamma})=\min_{Q \in \mathcal{Q}_2(\gamma)} D(Q\|P_2)\label{eq:min2},
\end{align}
where 
\begin{align}
\mathcal{Q}_1(\gamma)&=  \big\{Q\in \mathcal{P}(\mathcal{X}): D(Q\| P_1)-D(Q\|P_2) \geq  \gamma  \big\},\\
\mathcal{Q}_2(\gamma)&= \big\{Q\in \mathcal{P}(\mathcal{X}): D(Q\| P_1)-D(Q\|P_2) \leq  \gamma    \big\}.
\end{align}
The minimizing distribution in \eqref{eq:min1}, \eqref{eq:min2} is the tilted distribution
\begin{equation}\label{eq:tilted}
Q_{\lambda}(x)= \frac{ P_{1}^{1-\lambda}(x) P_{2}^{\lambda}(x) } {\sum_{a \in \mathcal{X} }  P_{1}^{1-\lambda}(a) P_{2}^{\lambda}(a) }, ~~~0\leq\lambda \leq 1
\end{equation}
whenever $\gamma$ satisfies $-D(P_1\|P_2) \leq \gamma  \leq D(P_2\|P_1)$.
In this case, $\lambda$ is the solution of
\begin{equation}\label{eq:KKTgamma} 
D(Q_{\lambda}\| P_1)-D(Q_{\lambda} \| P_2) = \gamma.
\end{equation}
Instead, if $\gamma<-D(P_1\|P_2)$, the optimal distribution in \eqref{eq:min1} is $Q_\lambda(x)= P_1(x)$ and $E_1(\phi_{\gamma})=0$, and if $\gamma>D(P_2\|P_1)$, the optimal distribution in \eqref{eq:min2} is $Q_\lambda(x)= P_2(x)$ and $E_2(\phi_{\gamma})=0$. 

Equivalently, the dual expressions of \eqref{eq:min1} and \eqref{eq:min2} can be derived by substituting the minimizing distribution \eqref{eq:tilted} into the Lagrangian yielding \cite{Blahut,Dembo} 
\begin{align}
E_1(\phi_{\gamma})&=\max_{\lambda \geq 0 } \lambda \gamma - \log \Big (  \sum_{x\in \mathcal{X}} P_1^{1-\lambda}(x) P_2^{\lambda}(x) \Big ),  \\
E_2(\phi_{\gamma})&=\max_{\lambda \geq 0 } -\lambda \gamma - \log \Big (  \sum_{x\in \mathcal{X}} P_1^{\lambda}(x) P_2^{1-\lambda}(x)  \Big ).  
\end{align}

The Stein regime is defined as the highest error exponent  under one hypothesis  when the  error probability under the other hypothesis is at most some fixed $ \epsilon \in (0,\frac{1}{2})$ \cite{Cover}
\begin{align}\label{eq:steindef}
E_2^{(\epsilon)}  \triangleq \sup \big \{E_2\in \mathbb{R}_{+}: \exists \phi , \exists n_0 \in \mathbb{Z}_+  \  \text{s.t.} \  \forall   n>n_0 \nonumber  \\
\epsilon_1 (\phi)\leq \epsilon  \quad \text{and} \quad \epsilon_2(\phi) \leq e^{-nE_2} \big \}.
\end{align} 
The optimal $E_2^{(\epsilon)}$, given by \cite{Cover}
\begin{equation}\label{eq:stein2}
E_2^{(\epsilon)} = D(P_1\|P_2),
\end{equation}
can be achieved by setting the threshold in \eqref{eq:LRTtype} to be ${\gamma} = -D(P_1\|P_2)+\frac{C_2}{\sqrt{n}}$, where $C_2$ is a constant that depends on distributions $P_1, P_2$ and $\epsilon$.

In this work, we revisit the above results in the case where the distributions used by the likelihood ratio test are not known precisely, and instead, fixed distributions $\hat P_1$ and $\hat P_2$ are used for testing. In particular, we find the error exponent tradeoff for fixed $\hat P_1$ and $\hat P_2$ and we study the worst-case tradeoff when the true distributions generating the observation are within a certain distance of the test distributions.
The literature in robust hypothesis testing is vast (see e.g., \cite{Huber,Kassam,poor2013introduction} and references therein). Robust hypothesis testing consists of designing tests that are robust to the inaccuracy of the distributions generating the observation. Instead, we study the error exponent tradeoff performance of the likelihood ratio test for fixed test distributions.

\section{Mismatched Likelihood Ratio Testing}
\label{sec:fixedHT}

Let $\hat{P}_1(x)$ and $\hat{P}_2(x)$ be the test distributions used in the likelihood ratio test with threshold $\hat{\gamma}$ given by
\begin{align}\label{eq:LRTtypeMM}
\hat{\phi}_{\hat{\gamma}}(\Th)= \mathbbm{1} \big\{ D(\Th\|\hat{P}_1)-D(\Th\|\hat{P}_2)    \geq \hat{\gamma} \big\}+1.
\end{align}
For simplicity, we assume that both $\hat P_1(x)>0$ and $\hat P_2(x)>0$ for each $x\in\Xc$. We are interested in the achievable error exponent of the mismatched likelihood ratio test, i.e., 
\begin{align}
\hat{E}_2(\hat{E}_1) \triangleq \sup& \big\{\hat{E}_2\in \mathbb{R}_{+}: \exists \hat{\gamma}  , \exists n_0 \in \mathbb{Z} _+ \  \text{s.t.} \  \forall  n>n_0 \nonumber \\
&\epsilon_1(\hat{\phi}_{\hat{\gamma}}) \leq e^{-n\hat{E}_1}  ~ \text{and} ~  \epsilon_2(\hat{\phi}_{\hat{\gamma}}) \leq e^{-n\hat{E}_2}\big\}.
\label{eq:opttestmism}
\end{align}
% Moreover, we are interested in conditions on threshold $\hat{\gamma}$ where both error exponents are strictly positive, and then we consider the mismatched LRT in the Stein's regime, i.e.,  we find the highest achievable error exponents using mismatched LRT.  

\begin{theorem}\label{thm:mismatchLRT}
		For fixed $\hat{P}_1, \hat{P}_2 \in \mathcal{P}(X)$ the optimal error exponent tradeoff in \eqref{eq:opttestmism} is given by		
	\begin{align}
	\hat{E}_1(\hat{\phi}_{\hat{\gamma}})&=\min_{Q \in \mathcal{\hat{Q}}_1 (\hat{\gamma})} D(Q\|P_1) \label{eq:LRTmis1}\\
	\hat{E}_2(\hat{\phi}_{\hat{\gamma}})&=\min_{Q \in \mathcal{\hat{Q}}_2(\hat{\gamma})} D(Q\|P_2)\label{eq:LRTmis2}
	\end{align}	
	where 
	\begin{align}
	\mathcal{\hat{Q}}_1(\hat{\gamma})&=  \big\{Q\in \mathcal{P}(\mathcal{X}): D(Q\| \hat{P}_1)-D(Q\|\hat{P}_2) \geq \hat{\gamma} \big \}, \label{eq:qhat1}\\
	\mathcal{\hat{Q}}_2(\hat{\gamma})&= \big\{Q\in \mathcal{P}(\mathcal{X}): D(Q\| \hat{P}_1)-D(Q\|\hat{P}_2) \leq  \hat{\gamma}   \big \}. \label{eq:qhat2}
	\end{align}
	The minimizing distributions in \eqref{eq:LRTmis1} and \eqref{eq:LRTmis2} are 
		\begin{equation}\label{eq:tiltedMM1}
\hat{Q}_{\lambda_1}(x)= \frac{ P_1(x)  \hat{P}_{1}^{-\lambda_1}(x) \hat{P}_{2}^{\lambda_1}(x) } {\sum_{a \in \mathcal{X} } P_1 (a) \hat{P}_{1}^{-\lambda_1}(a) \hat{P}_{2}^{\lambda_1}(a) },~~\lambda_1\geq0,
	\end{equation}
	\begin{equation}\label{eq:tiltedMM2}
	\hat{Q}_{\lambda_2}(x)= \frac{ P_2(x)  \hat{P}_{2}^{-\lambda_2}(x) \hat{P}_{1}^{\lambda_2}(x) } {\sum_{a \in \mathcal{X} } P_2 (a) \hat{P}_{2}^{-\lambda_2}(a) \hat{P}_{1}^{\lambda_2}(a) },~~\lambda_2\geq0
	\end{equation}
	respectively, where  $  \lambda_1$ is chosen so that	
	\begin{equation}\label{eq:KKTgamma1} 
	D(\hat{Q}_{\lambda_1}\|\hat{P}_1)-D(\hat{Q}_{\lambda_1} \| \hat{P}_2) = \hat{\gamma},
	\end{equation}
	whenever  
	$D(P_1\|\hat{P}_1)- D(P_1\|\hat{P}_2) \leq \hat{\gamma}$,
	    and otherwise, $\hat{Q}_{\lambda_1}(x)=P_1(x)$ and  $\hat{E}_1(\hat{\phi}_{\hat{\gamma}})=0$. Similarly,  $ \lambda_2 \geq 0$ is chosen so that
   \begin{equation}\label{eq:KKTgamma2} 
   D(\hat{Q}_{\lambda_2}\|\hat{P}_1)-D(\hat{Q}_{\lambda_2} \| \hat{P}_2) = \hat{\gamma},
   \end{equation}
   whenever
$   D(P_2 \|\hat{P}_1) -D(P_2 \|\hat{P}_2) \geq\hat{\gamma},$    and otherwise, $\hat{Q}_{\lambda_2}(x)=P_2(x)$ and  $\hat{E}_2(\hat{\phi}_{\hat{\gamma}})=0$. Furthermore,  the dual expressions for the type-\RNum{1} and type-\RNum{2} error exponents  are
\begin{align}\label{eq:dual}
	\hat{E}_1(\hat{\phi}_{\hat{\gamma}})&=\max_{\lambda \geq 0 } \lambda \hat{\gamma} - \log \Big (  \sum_{x\in \mathcal{X}} P_1(x) \hat{P}_1^{-\lambda}(x) P_2^{\lambda}(x) \Big ),  \\
	\hat{E}_2(\hat{\phi}_{\hat{\gamma}})&=\max_{\lambda \geq 0 } -\lambda \hat{\gamma} - \log \Big (  \sum_{x\in \mathcal{X}} P_1^{\lambda}(x) P_2(x) \hat{P}_2^{-\lambda}(x)  \Big ).  
	\end{align} 
\end{theorem}

%\begin{IEEEproof}
%The proof follows from Sanov's theorem and solving the KKT conditions. The duality expressions follow by substituting the minimizing distributions \eqref{eq:tiltedMM1} and \eqref{eq:tiltedMM2} into the Lagrangian. More details can be found in \cite{IZS-long}. 
%\end{IEEEproof}

\begin{remark}
For mismatched likelihood ratio testing, the optimizing distributions $\hat{Q}_{\lambda_1}, \hat{Q}_{\lambda_2}$ can be different, since the decision regions only depend on the mismatched distributions. However, if $\hat{P}_1, \hat{P}_2$ are tilted with respect to $P_1$ and $P_2$, then both $\hat{Q}_{\lambda_1}, \hat{Q}_{\lambda_2}$ are also tilted respect to $P_1$ and $P_2$. This implies the result in  \cite{Unnikrishnan}, where for any set of mismatched distributions $\hat{P}_1, \hat{P}_2$ that  are tilted with respect to generating distributions, the mismatched likelihood ratio test achieves the optimal error exponent tradeoff in \eqref{eq:tradefix}.
\end{remark}

\begin{theorem}\label{thm:stein}
	
	In the Stein regime, the mismatched likelihood ratio test achieves 
	\begin{equation}\label{eq:steinMM2}
	\hat{E}_2^{(\epsilon)}=\min_{Q \in \mathcal{\hat{Q}}_2(\hat{\gamma})} D(Q\|P_2),
	\end{equation}
	with threshold  
	\begin{equation}\label{eq:steinthresh2}
	    \hat{\gamma}=D(P_1\|\hat{P}_1) -D(P_1\|\hat{P}_2) +\frac{\hat{C}_2}{\sqrt{n}},
	\end{equation}
 and	$\hat{C}_2$ is a constant that depends on distributions $P_1,\hat{P}_1,\hat{P}_2$, and $\epsilon$.
\end{theorem} 

\begin{remark}	
	Note that since $P_1$ satisfies the constraint in \eqref{eq:steinMM2} then $\hat{E}_2^{(\epsilon)} \leq {E}_2^{(\epsilon)}$. In fact, if $\hat{P}_1, \hat{P}_2$ are tilted respect to $P_1, P_2$ then this inequality is met with equality. Moreover, it is easy to find a set of data and test distributions where  $\hat{E}_2^{(\epsilon)} < {E}_2^{(\epsilon)}$. 

\end{remark}

%%%%%%%%%%%%%%%%%%%%%%%%%%%%%%%%%%%%%%%%%
%%%%%%%%%%%%%%%%%%%%%%%%%%%%%%%%%%%%%%%%%

\section{Mismatched Likelihood Ratio \\Testing with Uncertainty}

In this section, we analyze the worst-case  error exponents tradeoff when the actual distributions $P_1, P_2$ are close to the mismatched test distributions $\hat{P}_1$ and $\hat{P}_2$. More specifically,
\begin{equation}\label{eq:ball2}
P_1 \in \mathcal{B}(\hat P_1,R_1),~~ P_2 \in \mathcal{B}(\hat P_2, R_2)
\end{equation}
where the $D$-ball
\begin{equation}\label{eq:ball}
\mathcal{B}( Q,R)= \big\{P\in\Pc(\Xc): D(Q\|P) \leq R \big\}
\end{equation} 
is a ball centered at distribution $Q$ containing all distributions whose relative entropy is smaller or equal than radius $R$. This model was used in robust hypothesis testing in \cite{Levy}.  Figure \ref{fig:mismatch} depicts the mismatched probability distributions and the mismatched likelihood ratio test as a hyperplane dividing the probability space into the two decision regions. 

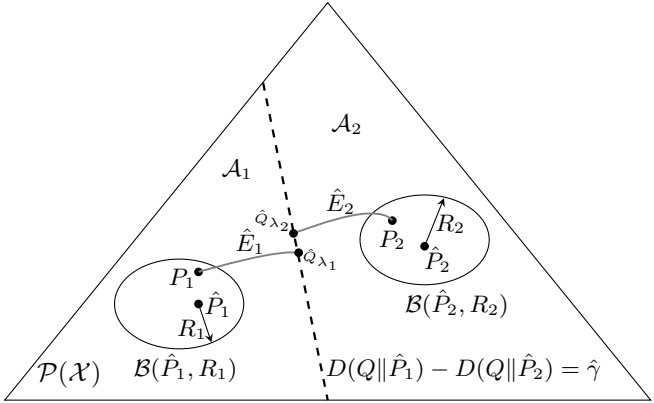
\begin{figure}[!h]
	\centering
	\begin{tikzpicture}[scale=0.85]
	\draw (5,1.2) -- (0,-5) -- (10,-5) --(5,1.2)  ;
	\draw [line width=0.3mm, dashed] (4,-0.05) --   (5.,-5)  ;
	\node at (7.1,-4.5) {\small $D(Q\|\hat{P}_1)-D(Q\|\hat{P}_2) =  \hat{\gamma}$};
%	\draw [cyan, xshift=4cm] plot [smooth, tension=0.5] coordinates { (-2,-2.5) (-0.5,-2.4) (1.5,-2.6) (4,-2.5)};
	\node[draw,circle,inner sep=1pt,fill] at (3,-3.5) {};
	\node at (3.3,-3.5) {\small $ \hat{P}_1$};
	\node[draw,circle,inner sep=1pt,fill] at (6.5,-2.6) {};
	\node at (6.7,-2.8) {\small $\hat{P}_2$};
	\draw (2.7,-3.5) ellipse (1cm and 0.7cm);
	\draw (6.5,-2.5) ellipse (1 cm and 0.7 cm);
	\node at (1,-4.6) {$\mathcal{P}({\mathcal{X}})$};
	\draw [->,>=stealth] (6.5,-2.6) -- (6.8,-1.85);
	\draw [->,>=stealth] (3,-3.5) -- (3.2,-4.1);
	\node at (2.9,-3.9) {\small $R_1$};
	\node at (6.9,-2.25) {\small $R_2$};
	\node[draw,circle,inner sep=1pt,fill] at (3,-3) {};
	\node[draw,circle,inner sep=1pt,fill] at (6,-2.2) {};
	\draw [line width=0.25mm, gray]plot [smooth, tension=1] coordinates{(3,-3)  (4,-2.75) (4.55,-2.7)} ;
	\draw [line width=0.25mm, gray]plot [smooth, tension=1] coordinates{ (6,-2.2)  (5.5,-2.1)  (4.47,-2.4)} ;
    \node[draw,circle,inner sep=1pt,fill] at (4.55,-2.7) {};
    \node[draw,circle,inner sep=1pt,fill] at (4.47,-2.4) {};
   \node at (4.9,-2.8) {\tiny $\hat{Q}_{\lambda_1}$};
     \node at (4.15,-2.2) {\tiny $\hat{Q}_{\lambda_2}$};
	\node at (3.8,-2.52) {\small $\hat E_1$};
	\node at (5.2,-1.9) {\small $\hat E_2$};
    \node at (2.75,-3.1) {\small $P_1$};
	\node at (6,-2.5) {\small $P_2$};
    \node at (3.6,-1.4) {\small $\Ac_1$};
	\node at (5.3,-0.7) {\small $\Ac_2$};
	 \node at (2.8,-4.5) {\small $\mathcal{B}(\hat P_1,R_1)$};
	\node at (7,-3.5) {\small $\mathcal{B}(\hat P_2,R_2)$};
	\end{tikzpicture}
	\caption{Mismatched likelihood ratio test over distributions in $D$-balls. } 
	\label{fig:mismatch}
\end{figure}

We study the worst-case error-exponent performance of mismatched likelihood ratio testing when the distributions generating the observation fulfill \eqref{eq:ball2}. 
In particular, we are interested in the least favorable distributions $P_1^L, P_2^L$ in  $\mathcal{B}(\hat P_1,R_1), \mathcal{B}(\hat P_2,R_2)$, i.e., the distributions achieving the lowest error exponents $\hat{E}^L_1(R_1), \hat{E}^L_2(R_2)$.

\begin{theorem}\label{thm:worstmismatch}
	For every  $ R_{1}, R_2 \geq 0$ let the least favorable exponents  $\hat{E}^L_1(R_1), \hat{E}^L_2(R_2)$ defined as
	\begin{align}
\hat{E}^L_1(R_1)&=\min_{P_1 \in \mathcal{B}(\hat P_1,R_1)    } \ \ \min_{Q \in \mathcal{\hat{Q}}_1 (\hat{\gamma})} D(Q\|P_1),\label{eq:MMlower1}\\
\hat{E}^L_2(R_2)&=\min_{P_2 \in  \mathcal{B}(\hat P_2,R_2)} \ \ \min_{ Q \in \mathcal{\hat{Q}}_2 (\hat{\gamma})  } D(Q\|P_2), \label{eq:MMlower2}
	\end{align}
where $\mathcal{\hat{Q}}_1 (\hat{\gamma}), \mathcal{\hat{Q}}_2 (\hat{\gamma}) $ are defined in \eqref{eq:qhat1}, \eqref{eq:qhat2}. Then, for any distribution pair $P_1 \in\mathcal{B}(\hat P_1,R_1), P_2 \in \mathcal{B}(\hat P_2,R_2)$, the corresponding error exponent pair $(\hat{E}_1, \hat{E}_2)$  satisfies
	\begin{equation}\label{eq:LU1}
\hat{E}^L_1(R_1)  \leq		\hat{E}_1(\hat{\phi}_{\hat{\gamma}})  \text{,} \quad \hat{E}^L_2(R_2)  \leq		\hat{E}_2(\hat{\phi}_{\hat{\gamma}}).
	\end{equation}
 Furthermore, the optimization problem in \eqref{eq:MMlower1} is convex with optimizing distributions 
	\begin{align}\label{eq:lowerworstKKT1}
	{Q}^L_{\lambda_1}(x)&= \frac{ P^L_1(x) \hat{P}_{1}^{-\lambda_1}(x) \hat{P}_{2}^{\lambda_1}(x) } {\sum_{a \in \mathcal{X} } 	P^L_1(a) \hat{P}_{1}^{-\lambda_1}(a) \hat{P}_{2}^{\lambda_1}(a) },\\
	P^L_1(x)&=\beta_1 Q^L_{\lambda_1}(x) + (1-\beta_1) \hat{P}_1(x),	\label{eq:lowerworstKKT11}
	\end{align}
	where $\lambda_1 \geq 0, 0\leq \beta_1 \leq 1 $  are chosen such that	
	\begin{align}\label{eq:condballMM}
D(Q^L_{\lambda_1}\|\hat{P}_1)-D(Q^L_{\lambda_1} \| \hat{P}_2) &=\hat{ \gamma},\\
D(\hat{P}_1\| P^L_1) &= R_{1},\label{eq:condballMM2}
\end{align}
	when
	\begin{equation}\label{eq:gammaworsL}
     \max_{P_1 \in\mathcal{B}(\hat P_1,R_1) }  D(P_1\| \hat{P}_1) - D(P_1\|\hat{P}_2)  \leq \hat{\gamma}.
	\end{equation}
	%where	
	%\begin{align}
	%\underline{\hat{\gamma}} =\max_{P_1 : D(\hat{P}_1 \| P_1 )  \leq R_{1} }  D(P_1\| \hat{P}_1) - D(P_1\|\hat{P}_2) 
	%\end{align}
	Otherwise, we can find a least favorable distribution $ P^{L}_1 \in \mathcal{B}(\hat P_1,R_1)$ such that  $\hat{E}_1(\hat{\phi}_{\hat{\gamma}})$ for this distribution is $\hat{E}_1(\hat{\phi}_{\hat{\gamma}})=0$.  Similarly,  the optimization  \eqref{eq:MMlower2} is convex with optimizing distributions 		
	\begin{align}\label{eq:lowerworstKKT2}
	Q^L_{\lambda_2}(x)&= \frac{ 	P^L_2(x) \hat{P}_{2}^{-\lambda_2}(x) \hat{P}_{1}^{\lambda_2}(x) } {\sum_{a \in \mathcal{X} } 	P^L_2(a) \hat{P}_{2}^{-\lambda_2}(a) \hat{P}_{1}^{\lambda_2}(a) },\\
	P^L_2(x)&=\beta_2 Q^L_{\lambda_2}(x) + (1-\beta_2) \hat{P}_2(x),
	\end{align}
	where $ \lambda_2 \geq 0, 0\leq \beta_2 \leq 1 $  are chosen such that	
	\begin{align}
	D(Q^L_{\lambda_2}\|\hat{P}_2)-DQ^L_{\lambda_2} \| \hat{P}_1) &=\hat{ \gamma},\\
	D(\hat{P}_2\| P^L_2) &= R_{2},
	\end{align}
	whenever,	
	\begin{equation}\label{eq:gammaworstU}
	\min_{P_2 \in\mathcal{B}(\hat P_2,R_2)}  D(P_2\| \hat{P}_1) - D(P_2\|\hat{P}_2) \geq \hat{\gamma}.
	\end{equation}
	%where
	%\begin{align}
	%\bar{\hat{\gamma}} =\min_{P_2 : D(\hat{P}_2 \| P_2 )  \leq R_{2} }  D(P_2\| \hat{P}_1) - D(P_2\|\hat{P}_2) 
	%\end{align}
	Otherwise, we can find a distribution $ P^{L}_2 \in \mathcal{B}(\hat P_2,R_2)$ such that  $\hat{E}_2(\hat{\phi}_{\hat{\gamma}})$ for this distribution is $\hat{E}_2(\hat{\phi}_{\hat{\gamma}})=0$. 
	\end{theorem}

The worst-case achievable error exponents of mismatched likelihood ratio testing for data distributions in a $D$-ball are essentially the minimum relative entropy between two sets of probability distributions. Specifically, the minimum relative entropy $\Bc(\hat P_1,R_1)$ and $\hat \Qc_2(\hat \gamma)$ gives $\hat{E}_1^L(R_1)$, and similarly for $\hat{E}_2^L(R_2)$.

%Theorem \ref{thm:worstmismatch} provides an expression for the worst-case achievable exponents over the family of distributions when mismatched LRT is used as well as conditions \eqref{eq:gammaworsL} \eqref{eq:gammaworstU} on the threshold $\hat{\gamma}$ such that the worst-case exponents are positive. Also, from Theorem \eqref{thm:stein} we can find the highest achievable exponents under each hypothesis when the probability of error under alternative hypothesis is bounded from unity for all possible distributions generating the data in the family. For this,  

%%%%%%%%%%%%%%%%%%%%%%%%%%%%%
%%%%%%%%%%%%%%%%%%%%%%%%%%%%%
\section{Mismatched Likelihood Ratio \\Testing Sensitivity}

In this section, we study how the worst-case error exponents $(\hat{E}^L_1, \hat{E}^L_2)$ behave when the $D$-ball radii $R_1,R_2$ are small. In particular, we  derive a Taylor series expansion of the worst-case error exponent. This approximation can also be interpreted as the worst-case sensitivity of the test, i.e., how does the test perform when actual distributions are very close to the mismatched distributions.    

\begin{theorem}\label{thm:lowerworst}
For every 	$ R_i \geq 0$,  $\hat{P}_i \in \mathcal{P}(\mathcal{X})$ for $i=1,2$, and 
\begin{equation}\label{eq:threshcodsen}
-D(\hat{P}_1\|\hat{P}_2) \leq \hat{\gamma} \leq D(\hat{P}_2\|\hat{P}_1),
\end{equation}
  we have	
		\begin{equation}\label{eq:worstapprox}
	   \hat{E}^L_i (R_i) = E_i(\hat{\phi}_{\hat{\gamma}}) - S_i(\hat{P}_1,\hat{P}_2,\hat{\gamma}) \sqrt{R_i}+ o(\sqrt{R_i}),  
	\end{equation}
	where	
%\begin{equation}
%S(\hat{P}_1,\hat{P}_2,\hat{\gamma}) = \frac{\sum_{x\in \mathcal{X}} \frac{{\hat{Q}_{\lambda}}^2(x)}{\hat{P_1}(x)}  -\frac{1}{|\mathcal{X}|} } {  \sqrt{ \sum_{x\in \mathcal{X}} \%frac{{\hat{Q}_{\lambda}}^2(x)}{\hat{P_1}(x)}  -\frac{2}{|\mathcal{X}|} +\frac{1}{|\mathcal{X}|^2} } },    
%\end{equation}
\begin{equation}\label{eq:sensitivity}
S_i^2(\hat{P}_1,\hat{P}_2,\hat{\gamma}) =2  \text{Var}_{\hat{P}_i} \Bigg(\frac{\hat{Q}_{\lambda}(X)}{\hat{P}_i(X)}  \Bigg) 
\end{equation}
and $\hat{Q}_{\lambda}(X)$ is the minimizing distribution in \eqref{eq:tilted} for test $\hat{\phi}_{\hat{\gamma}}$. 
\end{theorem}

\begin{lemma}\label{lem:varderivative}
	For every $\hat{P}_1,\hat P_2 \in \mathcal{P}(\mathcal{X})$,  and $\hat{\gamma}$ satisfying \eqref{eq:threshcodsen}
	\begin{align}
	\frac{\partial }{\partial \hat{\gamma}}S_1(\hat{P}_1,\hat{P}_2,\hat{\gamma}) \geq 0, ~~ 	\frac{\partial }{\partial \hat{\gamma}}S_2(\hat{P}_1,\hat{P}_2,\hat{\gamma}) \leq 0.
	\end{align}
\end{lemma}

This lemma shows that $S_1(\hat{P}_1,\hat{P}_2,\hat{\gamma})$  is a non-decreasing  function of $\hat{\gamma}$, i.e., as $\hat{\gamma}$ increases from $-D(\hat{P}_1\|\hat{P}_2) $ to $D(\hat{P}_2\|\hat{P}_1)$, the worst-case exponent $\hat E_1^L(R_1)$ becomes more sensitive to mismatch with  likelihood ratio testing. Conversely, $S_2(\hat{P}_1,\hat{P}_2,\hat{\gamma})$  is a non-increasing  function of $\hat{\gamma}$, i.e., as $\hat{\gamma}$ increases from $-D(\hat{P}_1\|\hat{P}_2) $ to $D(\hat{P}_2\|\hat{P}_1)$, the worst-case exponent $\hat E_2^L(R_2)$ becomes less sensitive (more robust) to mismatch with  likelihood ratio testing. Moreover, when $\lambda=\frac{1}{2}$, we have 	
		\begin{equation}
		\hat{Q}_{\frac{1}{2}}(x)=\frac{\sqrt{\hat{P}_1(x) \hat{P}_2(x) }}{ \sum_{a \in \mathcal{X} } \sqrt{\hat{P}_1(a) \hat{P}_2(a) }    },
		\end{equation}
		and then $S_1(\hat{P}_1,\hat{P}_2,\hat{\gamma})=S_2(\hat{P}_1,\hat{P}_2,\hat{\gamma})$. In addition,  $\hat{Q}_{\frac{1}{2}}$ minimizes ${E}_1(\hat \phi_{\hat \gamma}) + {E}_2(\hat \phi_{\hat \gamma})$  yielding \cite{Veeravalli}
		\begin{align}
		{E}_1(\hat \phi_{\hat \gamma}) + {E}_2(\hat \phi_{\hat \gamma})&= \min_{Q \in \mathcal{P}(\mathcal{X})} D(Q\|\hat{P}_1) +  D(Q\|\hat{P}_2)\\& = 2 B(\hat{P}_1,\hat{P}_2) 
				\end{align}
where $ B(\hat{P}_1,\hat{P}_2)$ is the Bhattacharyya distance between the mismatched distributions $\hat{P}_1$ and $\hat{P}_2$. This suggests that having equal sensitivity (or robustness) for both hypotheses minimizes the sum of the exponents.

%From this result we can 

\begin{example}
When $\gamma=0$ the likelihood ratio test becomes the maximum-likelihood test, which is known to achieve the lowest average probability of error in the Bayes setting for equal priors. For fixed priors $\pi_1,\pi_2$, the error probability in the Bayes setting is
$\bar\epsilon= \pi_1\epsilon_1 +\pi_2\epsilon_2$,
 resulting in the following error exponent \cite{Cover}
\begin{equation}
\bar E= \lim_{n\rightarrow \infty} \frac{1}{n} \log \bar \epsilon = \min \{E_1,E_2\}.
\end{equation}
Consider $\hat{P}_1 =\text{Bern}(0.1)$ ,    $\hat{P}_2 =\text{Bern}(0.8)$. Also, assume $R_1=R_2=R$. Figure \ref{fig:worstRsen} shows the worst-case error exponent in the Bayes setting given by $\min \{\hat E_1^L,\hat E_2^L\}$ by solving \eqref{eq:MMlower1} and \eqref{eq:MMlower2} as well as  $\min \{\tilde E_1^L,\tilde E_2^L\}$ 
using the approximation in   \eqref{eq:worstapprox}. We can see that the approximation is good for small $R$. Moreover, it can be seen that error exponents are very sensitive to mismatch for small $R$, i.e., the slope of the worst-case exponent goes to infinity as $R$ approaches to zero. 

\begin{figure}[!h]
	\centering
	\includegraphics[width=.5\textwidth]{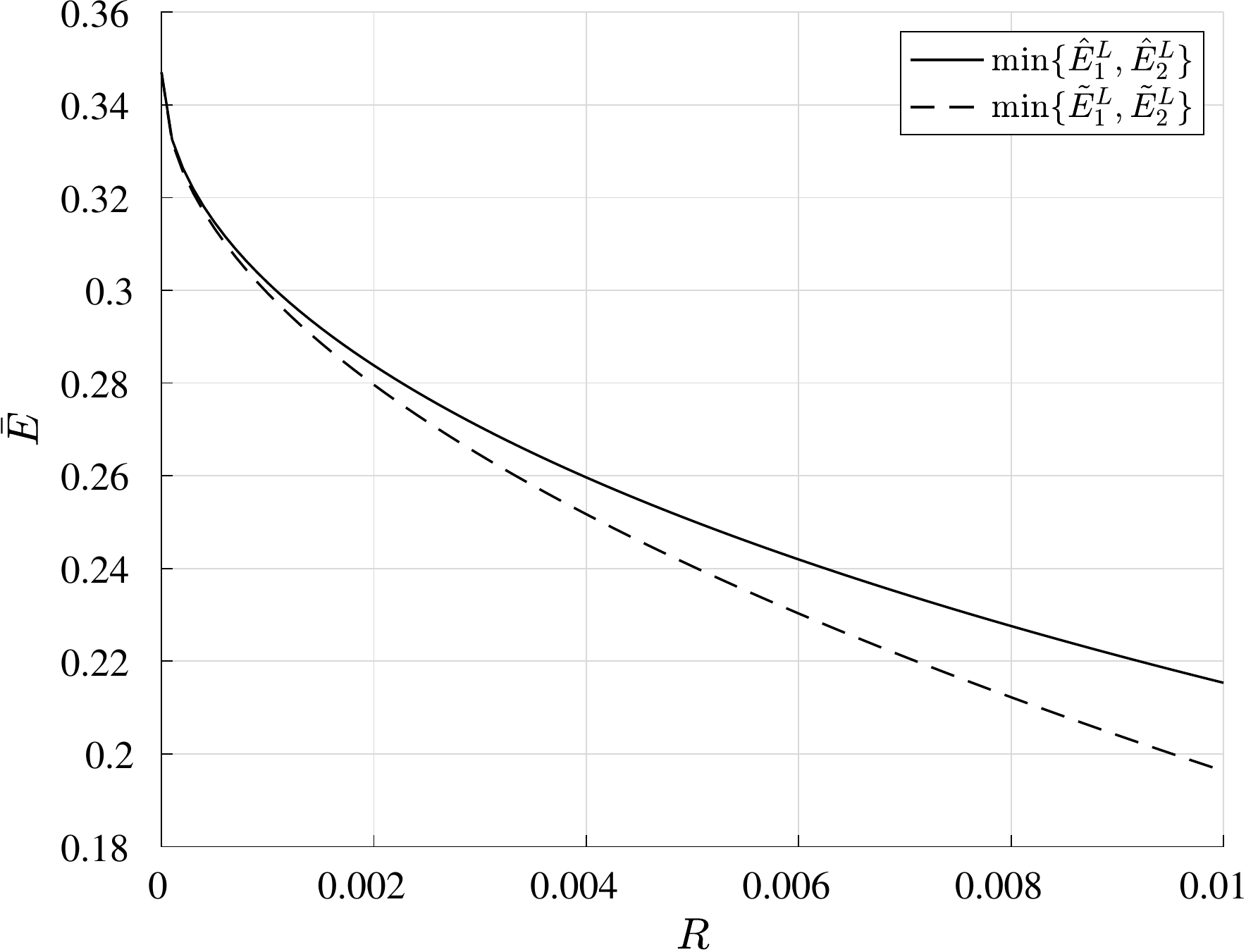}
	\caption{Worst-case achievable Bayes error exponent.} 
	\label{fig:worstRsen}
\end{figure}

\end{example}

\section*{Appendix}

	\subsection{Proof of Theorem \ref{thm:mismatchLRT}}
		We show the result  for $	\hat{E}_1(\hat{\phi}_{\hat{\gamma}})$ and similar steps are valid for $\hat{E}_2(\hat{\phi}_{\hat{\gamma}})$. The type-\RNum{1} probability of error can be written as 
		\begin{align}	
		\hat{\epsilon}_1(\hat{\phi}_{\hat{\gamma}})  = \sum_{\substack{\bx\in\Xc^n\\D(\Th\|\hat{P}_1)-D(\Th\|\hat{P}_2)  \geq \hat{\gamma}}} P_{1}^n(\xv).
		%\hat{\epsilon}_1(\hat{\phi}_{\hat{\gamma}})=P_1 \Big ( \bx:D(\Th\|\hat{P}_1)-D(\Th\|\hat{P}_2)  \geq \hat{\gamma}  \Big ).
		\label{eq:tailmismatch}
		\end{align}
		Applying Sanov's Theorem  to \eqref{eq:tailmismatch} to get  \eqref{eq:LRTmis1} is immediate. The optimization problem in  \eqref{eq:LRTmis1}  consists of the minimization of a convex function over linear constraints. Therefore, the KKT conditions are also sufficient \cite{Boyd}. Writing the Lagrangian, we have
		\begin{align}\label{eq:lagrangeMM}
		L(Q,\lambda,\nu)= &D(Q\|P_1) + \lambda \big ( D(Q\|\hat{P}_2)-D(Q\|\hat{P}_1) +\hat{\gamma} \big ) \nonumber \\
		&+\nu \Big ( \sum_{x\in \mathcal{X}} Q(x)-1 \Big).
		\end{align}
		Differentiating with respect to $Q(x)$ and setting to zero we have
		\begin{equation}\label{eq:lagrangeder}
		1+\log \frac{Q(x)}{P_1(x)} +\lambda \log \frac{\hat{P}_1(x)}{\hat{P}_2(x)} + \nu=0.
		\end{equation} 
		Solving equations \eqref{eq:lagrangeder} for every $x\in\Xc$ we obtain \eqref{eq:tiltedMM1}. Moreover, from the complementary slackness condition if \cite{Boyd}
		\begin{equation}\label{eq:threshcondition}
		D(P_1\|\hat{P}_1)- D(P_1\|\hat{P}_2) \leq \hat{\gamma},
		\end{equation}
		then \eqref{eq:KKTgamma1} should hold. Otherwise, if \eqref{eq:threshcondition} does not hold then $\lambda$ in \eqref{eq:lagrangeder} should be zero and hence $\hat{Q}_{\lambda_1}=P_1$, $\hat{E}_1(\hat{\phi}_{\hat{\gamma}})=0$.  Finally, substituting the minimizing distribution $\hat{Q}_{\lambda_1}$ \eqref{eq:tiltedMM1} into \eqref{eq:lagrangeMM} we get the dual expression
		\begin{equation}\label{eq:lagrange}
		g(\lambda)= \lambda \hat{\gamma} - \log \Big (  \sum_{x\in \mathcal{X}} P_1 \hat{P}_1^{-\lambda}(x) P_2^{\lambda}(x) \Big ).  
		\end{equation}
		Since the optimization problem in \eqref{eq:LRTmis1} is convex, then the duality gap is zero \cite{Boyd}, and this proves the \eqref{eq:dual}. 
%		Alternatively, we can show \eqref{eq:gammaL} using the dual expression. Note that the second term of RHS in \eqref{eq:dual} is convex respect to $\lambda$. Then if the slope of $\log \Big (  \sum_{x\in \mathcal{X}} P_1 \hat{P}_1^{-\lambda}(x) P_2^{\lambda}(x) \Big )$ at the origin is greater than the slope of the line $\lambda \hat{\gamma}$, i.e.,
%		
%		\begin{align}
%		\hat{\gamma} & \leq   \frac{d}{d \lambda }	\log \Big (  \sum_{x\in \mathcal{X}} P_1(x) \hat{P}_1^{-\lambda}(x) P_2^{\lambda}(x) \Big ) \Big |_{\lambda=0}\\
%		&=D(P1\|\hat{P}_1) -D(P_1\hat{P}_2),   
%		\end{align}
%		then $\hat{E}_1(\hat{\phi}_{\hat{\gamma}})=0$, which is the condition in \eqref{eq:gammaL}.
		
	\subsection{Proof of Theorem \ref{thm:stein}}
     
     First, notice that  $\hat{E}_2(\hat{\phi}_{\hat{\gamma}})$ is a non-increasing function of $\hat{\gamma}$  since for every $\hat{\gamma}_1 \leq \hat{\gamma}_2 $  we have
     \begin{equation}
     \mathcal{\hat{Q}}_2({\hat{\gamma}_1}) \subset   \mathcal{\hat{Q}}_2({\hat{\gamma}_2}),
     \end{equation}
     hence 
     \begin{equation}
     \hat{E}_2(\hat{\phi}_{\hat{\gamma}_2}) \leq \hat{E}_2(\hat{\phi}_{\hat{\gamma}_1}).
     \end{equation}
      Therefore, in the Stein's regime we are looking for the smallest threshold such that $\limsup_{n\rightarrow \infty}  \hat{\epsilon}_1 (\hat{\phi}_{\hat{\gamma}}) \leq \epsilon$. Let
    \begin{equation}\label{eq:steinthresh}
    \hat{\gamma}= D(P_1\|\hat{P}_1) - D(P_1\|\hat{P}_2) -   \sqrt {\frac{ V(P_1, \hat{P}_1,\hat{P}_2) } {n} }    \Phi^{-1}(\epsilon),
    \end{equation}	
    where 
    \begin{align}
  &V(P_1, \hat{P}_1,\hat{P}_2)= \text{Var}_{P_1}\bigg(   \log \frac{\hat{P}_1}{\hat{P}_2} \bigg ) \nonumber\\
 & = \sum_{x\in \mathcal{X}} P_1(x)  \bigg ( \log \frac{\hat{P}_1}{\hat{P}_2} \bigg )^2 - 
   \big ( D(P_1\|\hat{P}_2) - D(P_1\|\hat{P}_1)   \big)^2,
    \end{align}
    and $\Phi^{-1}(\epsilon)$ is the inverse cumulative distribution function of a zero-mean unit-variance Guassian random variable. For such $\hat{\gamma}$, the type-\RNum{1}  error probability of the mismatched likelihood ratio test is
     \begin{align}	
    \hat{\epsilon}_1(\hat{\phi}_{\hat{\gamma}})=&\Pp_1 \Bigg [\frac{1}{n} \sum_{i=1}^{n} \log \frac{\hat{P}_1(X_i)}{\hat{P}_2(X_i)}  \leq 
     D(P_1\|\hat{P}_2) - D(P_1\|\hat{P}_1) \nonumber\\
      &+   \sqrt {\frac{ V(P_1, \hat{P}_1,\hat{P}_2) } {n} }    \Phi^{-1}(\epsilon) \Bigg ].
    \end{align}
    Observe that $D(P_1\|\hat{P}_2) - D(P_1\|\hat{P}_1) =  \mathbb{E}_{P_1} \Big [ \log \frac{\hat{P}_1(X)}{\hat{P}_2(X)} \Big ]$. Let $  \hat{S}_n=\frac{1}{n}\sum_{i=1}^n \hat\imath(x_i)$, where $\hat\imath(x_i)=\log \frac{\hat{P}_1(x_i)}{\hat{P}_2(x_i)}$.  Letting $Z$ be a zero-mean unit-variance Guassian random variable, then, by the central limit theorem we have
    \begin{align}
    &\limsup_{n\rightarrow \infty}  \hat{\epsilon}_1 (\hat{\phi}_{\hat{\gamma}})  &\notag\\
    &= \limsup_{n\rightarrow \infty} \Pp_1\Bigg [ \frac{ \sqrt{n} \big ( \hat{S}_n- \mathbb{E}_{P_1} [\hat\imath(X)] \big )}{\sqrt {V(P_1, \hat{P}_1,\hat{P}_2) }  }   \leq    \Phi^{-1}(\epsilon)     \Bigg]\\
   &=\Pp\big [Z \leq \Phi^{-1}(\epsilon) \big]\\
   &= \epsilon. 
    \end{align} 
    Therefore, asymptotically, the type-\RNum{1} error probability of mismatched likelihood ratio test with $\hat{\gamma}$ in \eqref{eq:steinthresh} is equal to $\epsilon$. 
    
    Next, we need to show that for any threshold $\hat{\gamma}$ and $\varepsilon>0$ such that
    \begin{equation}\label{eq:limsupthresh}
       \limsup_{n \rightarrow \infty} \hat{\gamma} +\varepsilon\leq D(P_1\|\hat{P}_1) - D(P_1\|\hat{P}_2),
    \end{equation}
    the type-\RNum{1} probability of error tends to $1$ as the number of observation approaches infinity, which implies that $D(P_1\|\hat{P}_1) - D(P_1\|\hat{P}_2)$ is the lowest possible threshold that meets the constraint $\limsup_{n\rightarrow \infty}  \hat{\epsilon}_1 (\hat{\phi}_{\hat{\gamma}}) \leq \epsilon$. The corresponding $\hat E_2(\hat\phi_{\hat\gamma})$ is this highest type-\RNum{2} exponent that meets the constraint.
    In order to show this, define the following sets 
    \begin{align}
    \mathcal{E}_{\delta} = \Big \{ \bx\in\Xc^n:\, &\| \Th(x) -P_1(x)\|_\infty   < \delta \Big     \},\\   
    \Dc = \big\{  \bx\in\Xc^n:\, & \big|D(\Th\|\hat{P}_1)-D(\Th\|\hat{P}_2)\\  
    -&D(P_1\|\hat{P}_1)+D(P_1\|\hat{P}_2) \big | <  \varepsilon \big\}, \nonumber\\
    \bar \Dc = \big\{  \bx\in\Xc: \,&D(\Th\|\hat{P}_1)-D(\Th\|\hat{P}_2) \nonumber\\
    -& D(P_1\|\hat{P}_1)+D(P_1\|\hat{P}_2) \geq - \varepsilon \big\}.
    \end{align}
    where $\|.\|_{\infty}$ is the norm infinity. From the continouity of $D(.\|\hat{P})$ we have that for any $\varepsilon >0$ such that
    \begin{equation}\label{eq:epsilondelta}
    \big |D(\Th\|\hat{P}_1)-D(\Th\|\hat{P}_2) - D(P_1\|\hat{P}_1)+D(P_1\|\hat{P}_2) \big  | < \varepsilon. 
    \end{equation} 
    there exists $\delta>0$ such that for all $\Th$ satisfying
        \begin{equation}
    \| \Th(x) -P_1(x)\|_\infty   < \delta 
    \end{equation}
     \eqref{eq:epsilondelta} holds. Therefore, when  \eqref{eq:limsupthresh} holds 
\begin{align}
      \liminf_{n\rightarrow \infty} \epsilon_1 (\hat{\phi}_{\hat{\gamma}} ) \geq&  \liminf_{n\rightarrow \infty}  \sum_{x\in\bar \Dc} P_1^n(\xv)\\
     \geq  &\liminf_{n\rightarrow \infty}   \sum_{x\in \Dc} P_1^n(\xv).   
       \end{align}
 Now from the continuity argument, there exists a $\delta$ such that 
 \begin{equation}
  \sum_{x\in\Dc} P_1^n(\xv) \geq   \sum_{x\in\Ec_{\delta}} P_1^n(\xv).
 \end{equation}
Set $\delta_n=\sqrt{\frac{\log n}{n}}$. Thus, for sufficiently large $n$, $\delta_n \leq \delta$, Therefore, we have 
     \begin{align}
       \liminf_{n\rightarrow \infty} \epsilon_1 (\hat{\phi}_{\hat{\gamma}} )  &\geq \liminf_{n \rightarrow \infty} \sum_{\xv\in\Ec_{\delta_n}} P_1^n(\xv)\\
  %     
 %      \Pp_1 \big [\bx: \| \Th(x) -P_1(x)\|_\infty   < \delta_n \big ) \\
 %& = \liminf_{n \rightarrow \infty} 1-\sum_{\xv\in\Ec_{\delta_n}^c} P_1^n(\xv)\\
&\geq \lim_{n \rightarrow \infty} 1- \frac{2|\mathcal{X}|}{n}\\
&=1.
     \end{align}
where the last step is by Hoeffding's inequality \cite{Hoeffdingineq} and union bound. Therefore, for any $\hat{\gamma} < D(P_1\|\hat{P}_1) - D(P_1\|\hat{P}_2)$ type-\RNum{1} error goes to unity which concludes the theorem.

%%%%%%%%%%%%%%%%%%%%%%%%%%%%
	
	\subsection{Proof of Theorem \ref{thm:worstmismatch}}
	
		We show the result  under the first hypothesis and similar steps are valid under the second hypothesis. For every $P_1$ the achievable type-\RNum{1} is error exponent  	$\hat{E}_1(\hat{\phi}_{\hat{\gamma}})$  does not depend on $P_2$ therefore, \eqref{eq:MMlower1} is a lower bound to $\hat{E}_1(\hat{\phi}_{\hat{\gamma}})$. Moreover, since the relative entropy is jointly convex, then \eqref{eq:MMlower1} is a convex optimization problem and the KKT conditions are also sufficient.  Writing the Lagrangian we have
		\begin{align}\label{eq:lagrange}
		L(Q,P_1,\lambda_1,\lambda_1', \nu_1, \nu_1')= D(Q\|P_1) + \lambda_1 \big( D(Q\|\hat{P}_2) \nonumber  \\
		-D(Q\|\hat{P}_1) +\hat{\gamma} \big) +  \lambda_1' \big ( D(\hat{P}_1\|P_1) -R_1\big ) \nonumber \\
		+ \nu_1 \Big ( \sum_{x\in \mathcal{X}} Q(x)-1\Big )+ \nu_1' \Big ( \sum_{x\in \mathcal{X}} P_1(x)-1\Big ).
		\end{align}
		Differentiating with respect to $Q(x)$ and  $P_1(x)$ and setting the derivatives to zero we have
		\begin{align}
		1+\log \frac{Q(x)}{P_1(x)} +\lambda_1 \log \frac{\hat{P}_1(x)}{\hat{P}_2(x)} + \nu_1&=0,\label{eq:lagrange1}\\
		-\frac{Q(x)}{P_1(x)}-\lambda_1' \frac{\hat{P}_1(x)}{P_1(x)}+\nu_1'&=0, \label{eq:lagrange2}
		\end{align} 
		respectively. Solving equations \eqref{eq:lagrange1}, \eqref{eq:lagrange2} for every $x\in\Xc$ and letting $\beta_1=\frac{1}{1+\lambda_1'}$  we obtain \eqref{eq:lowerworstKKT1} and \eqref{eq:lowerworstKKT11}. Moreover, from the complementary slackness condition \cite{Boyd} if  for all  $P_1$ in $\mathcal{B}(\hat{P}_1,R_1)$ the condition $D(P_1\|\hat{P}_1)- D(P_1\|\hat{P}_2) \leq \hat{\gamma}$ stated in Theorem \ref{thm:mismatchLRT} holds, then \eqref{eq:condballMM} and \eqref{eq:condballMM2} should hold. Otherwise,  if there exists a  $P^{L}_1$  in $\mathcal{B}(\hat P_1,R_1) $ such that $D(P_1^L\|\hat{P}_1)- D(P_1^L\|\hat{P}_2) \leq \hat{\gamma}$, then for this distribution $\hat{E}_1(\hat{\phi}_{\hat{\gamma}})=0$. Therefore, if conditon \eqref{eq:gammaworsL} holds for all $P_1$ in the $D$-ball  $\hat{E}^L_1(R_1) >0$, otherewise $\hat{E}^L_1(R_1) =0$.

		\subsection{Proof of Theorem \ref{thm:lowerworst}}
		
			We show the result under the first hypothesis, and similar steps are valid for the second hypothesis.  Consider the first minimization in \eqref{eq:MMlower1} over $Q$, i.e.,
			\begin{equation}\label{eq:LRTmis1perturb}
			\hat{E}_1(\hat{\phi}_{\hat{\gamma}})= \min_{ Q \in \mathcal{\hat{Q}}_1 (\hat{\gamma})  } D(Q\|P_1).
			\end{equation}
			First, note that by assumption, $\hat P_1(x)>0$ for each $x\in\Xc$. Therefore, for any finite $R_1$, we have $P_1(x)>0$ for every  $P_1 \in \mathcal{B}(\hat P_1,R_1)$. Hence, for $P_1 \in \mathcal{B}(\hat P_1,R_1)$, the relative entropy $D(Q\|P_1)$ is continuous in both $Q, P_1$.  Moreover, the constraints in  \eqref{eq:LRTmis1perturb} are continuous with respect to $Q$ and also trivially with respect to $P_1$, since the constraints do not depend on $P_1$. Hence, the optimization in \eqref{eq:LRTmis1perturb} is minimizing a continuous function  over a compact set with continuous constraints.  Hence, by the maximum theorem \cite{Walker},
			% for any sequence $P_1^k \in \mathcal{B}_1(R_1)$ that  tends to $P_1$ then $    \hat{E}_1(\hat{\phi}_{\hat{\gamma}}|P_k)$ goes to $    \hat{E}_1(\hat{\phi}_{\hat{\gamma}})$, i.e.,   
			$\hat{E}_1(\hat{\phi}_{\hat{\gamma}})$ is  a continuous function of $P_1$  for all $P_1 \in \mathcal{B}(\hat P_1,R_1)$ with finite radius $R_1$. Therefore, by the envelope theorem\cite{Segal} we have
			\begin{equation}
			\frac{ \partial \hat{E}_1(\hat{\phi}_{\hat{\gamma}}) }{\partial P_1(x)}= -\frac{\hat{Q}_{\lambda}(x)}{P_1(x)}.  
			\end{equation} 
			Define the vectors
			\begin{align}
			\nabla \hat{E}_1 &= \bigg( -\frac{\hat{Q}_{\lambda}(x_1)}{\hat{P}_1(x_1)},\dotsc, -\frac{\hat{Q}_{\lambda}(x_{|\Xc|})}{\hat{P}_1(x_{|\Xc|})}\bigg)^T\\
			\thetav_{P_1}		 &= \big(P_1(x_1)-\hat{P}_1(x_1),\dotsc,P_1(x_{|\Xc|})-\hat{P}_1(x_{|\Xc|})\big)^T.
			\end{align}			
            Assuming the $\hat{E}^L_i(R_i)$ to be continuous we can apply the Taylor expansion to $\hat{E}_1(\hat{\phi}_{\hat{\gamma}})$ around $P_1=\hat{P}_1$ and we obtain
			\begin{align}\label{eq:linearapprox}
			\hat{E}_1(\hat{\phi}_{\hat{\gamma}})=E_1(\hat{\phi}_{\hat{\gamma}})  +     \thetav_{P_1}^{T}   \nabla \hat{E}_1     + o(\| \thetav_{P_1} \|_{\infty}).
			\end{align}
			%where $\delta P_1$, $\nabla \hat{E}_1 $ represent  vectors with entries $\{ P_1(x)-\hat{P}_1(x), x \in \mathcal{X} \}$, $\{ -\frac{\hat{Q}_{\lambda}(x)}{\hat{P}_1(x)}, x \in \mathcal{X} \}$ respectively. 
			%Also, since $D(Q\|P_1)$ is jointly convex, then minimization over $Q$ will result a convex function in $P_1$ (cite boyd). Therefore, the second order term is positive and can be replace by $| o(\| \delta P \|_{\infty})|$.
			By substituting the expansion \eqref{eq:linearapprox}   for  the first minimization in  \eqref{eq:MMlower1} we obtain
			\begin{equation}\label{eq:approxworst}
			\hat{E}^{L}_1(R_1) = \min_{P_1 \in \mathcal{B}(\hat P_1,R_1)   }  E_1(\hat{\phi}_{\hat{\gamma}})  +     \thetav_{P_1}^{T}   \nabla\hat{E}_1     + o(\|\thetav_{P_1} \|_{\infty}).
			\end{equation}
			
			Now, we further approximate the outer minimization constraint in \eqref{eq:MMlower1}. By approximating $D(\hat{P}_1 \| P_1 )$ we get \cite{Zheng}
			\begin{equation}\label{eq:KLapprox}
			D(\hat{P}_1 \| P_1 ) = \frac{1}{2}  \thetav_{P_1}^T \Jm(\hat{P}_1) \thetav_{P_1} + o (\| \thetav_{P_1} \|^2_{\infty}),
			\end{equation}
			where 
			\begin{equation}
			\Jm(\hat{P}_1)=\diag\bigg( \frac{1}{\hat{P}_1(x_1)},\dotsc,\frac{1}{\hat{P}_1(x_{|\Xc|})}\bigg)
			\end{equation}
			is the Fisher information matrix. Therefore,   \eqref{eq:approxworst} can be approximated as
			\begin{align}\label{eq:worstapproxopt}
			\hat{E}^{L}_1(R_1) &\approx   \tilde{E}^L_1 (R_1) \nonumber \\
			&\triangleq \min_{\substack{ \frac{1}{2} \thetav_{P_1}^T \Jm(\hat{P}_1) \thetav_{P_1}  \leq R_{1} \\ \onev^T\thetav_{P_1}=0}}  E_1(\hat{\phi}_{\hat{\gamma}})  +     \thetav_{P_1}^{T}   \nabla \hat{E}_1.    
			\end{align}	
			The optimization problem in \eqref{eq:worstapproxopt} is  convex and hence the KKT conditions are sufficient.  The corresponding Lagrangian is given by
			\begin{align}
			L(\thetav_{P_1}, \lambda,\nu) &= E_1(\hat{\phi}_{\hat{\gamma}})  +     \thetav_{P_1}^{T}   \nabla \hat{E}_1  \nonumber \\
			&+ \lambda \Big (\frac{1}{2}\thetav_{P_1}^T \Jm(\hat{P}_1) \thetav_{P_1}  - R_{1} \Big) +\nu ( \onev^T\thetav_{P_1} ).
			\end{align}
			Differentiating with respect to $\thetav_{P_1}$ and setting to zero, we have
			\begin{equation}\label{eq:KKTsen}
			\nabla \hat{E}_1 + \lambda \Jm(\hat{P}_1)\thetav_{P_1} +\nu \onev=0.
			\end{equation}
			Therefore,
			\begin{equation}\label{eq:deltaPsolution}
			\thetav_{P_1} =-\frac{1}{\lambda}   \Jm^{-1}(\hat{P}_1) \big (\nabla \hat{E}_1   +\nu \onev \big ).
			\end{equation}
			Note that if $\lambda=0$ then from \eqref{eq:KKTsen}    $  \nabla \hat{E}_1= -\nu \onev$ which cannot be true for thresholds satisfying \eqref{eq:threshcodsen} since $\hat{Q}_{\lambda} \neq \hat{P}_1$. Therefore, from the complementary slackness condition \cite{Boyd} the inequality constraint \eqref{eq:worstapproxopt} should be satisfied with equality.  By solving $\frac{1}{2}\thetav_{P_1}^T \Jm(\hat{P}_1) \thetav_{P_1} = R_{1}$ and  $\bold{1}^T\thetav_{P_1} =0 $  and substituting  $\lambda, \nu$ in \eqref{eq:deltaPsolution}, we obtain 
			\begin{equation}\label{eq:deltaP}
			\thetav_{P_1} =-\frac{ \psiv}{\sqrt{\psiv^T \Jm(\hat{P}_1)\psiv} }\sqrt{2R_1},
			\end{equation}
			where
			\begin{equation}	
			\psiv=	\Jm^{-1}(\hat{P}_1)\Bigg  (\nabla \hat{E}_1 -{\onev^T\Jm^{-1}(\hat{P}_1) \nabla \hat{E}_1  } \onev\Bigg ).
			\end{equation}
			 Substituding \eqref{eq:deltaP} into  \eqref{eq:approxworst} yields \eqref{eq:worstapprox}.

		%%%%%%%%%%%%%%%%%%%%%%%%%%%%%%%%%%%%%
		\subsection{Proof of Lemma \ref{lem:varderivative}}
		
		We show the result  under the first hypothesis and similar steps are valid under the second hypothesis. To prove the Theorem we need the following lemma. 
					\begin{lemma}\label{lem:convex}
					Consider the following optimization problem
					\begin{equation}
					E(\gamma)= \min_{ \mathbb{E}_Q [X] \geq \gamma } D(Q\|P).
					\end{equation} 
			Then	$E(\gamma)$ is convex in $\gamma$.
				\end{lemma}		
				
				\begin{proof}
					Let 
					\begin{equation}
					Q^{*}_{1} = \argmin_{ \mathbb{E}_Q [X] \geq \gamma_1} D(Q\|P) ~~ Q^{*}_{2} = \argmin_{ \mathbb{E}_Q [X] \geq \gamma_2} D(Q\|P).
					\end{equation}
					From the  convexity of the relative entropy, 	for any   $\alpha \in (0,1)$,

					\begin{align}
					D&\big(\alpha Q^*_1 +(1-\alpha) Q^*_2 \| P \big) \leq \alpha  D( Q^*_1  \| P) +(1-\alpha)  D( Q^*_2 \| P)\\
					&= \alpha  \min_{ \mathbb{E}_Q [X] \geq \gamma_1 } D(Q\|P) +(1-\alpha)  \min_{ \mathbb{E}_Q [X] \geq \gamma_2 } D(Q\|P).
					\end{align}			
					Furthermore,  since $Q^*_1, Q^*_2$  satisfy their correspending optimization constraints, then $\mathbb{E}_{Q^*_1}[X] \geq   \gamma_1$, $\mathbb{E}_{Q^*_2}[X]  \geq  \gamma_2$ , hence

					\begin{equation}
					\mathbb{E}_{\alpha Q^*_1 +(1-\alpha) Q^*_2}[X]  \geq \alpha \gamma_1+ (1-\alpha) \gamma_2.
					\end{equation}
					Therefore,  $\alpha Q^*_1 +(1-\alpha) Q^*_2$ satisfies the optimization constraint when $\gamma= \alpha \gamma_1 + (1-\alpha) \gamma_2$, then

					\begin{align}
					&\min_{	\mathbb{E}_{Q} [X] \leq  \alpha \gamma_1+(1-\alpha) \gamma_2} D(Q\|P) \leq D(\alpha Q^*_1 +(1-\alpha) Q^*_2 \| P)\\
					&\leq \alpha  \min_{ \mathbb{E}_Q [X] \geq \gamma_1 } D(Q\|P) +(1-\alpha)  \min_{ \mathbb{E}_Q [X] \geq \gamma_2 } D(Q\|P).
					\end{align}	 
					Hence $E(\gamma)$ is  convex in $\gamma$. 
					
				\end{proof}
				From above lemma we can show that $\lambda$  is a non-decreasing function of $\hat{\gamma}$. From the envelope theorem \cite{Segal}
     \begin{equation}
   \frac{\partial 	\hat{E}_1(\hat{\phi}_{\hat{\gamma}})  }{\partial \hat{\gamma}} = \lambda^*,
     \end{equation}
     where $\lambda^*$ is the optimizing $\lambda$ in \eqref{eq:tilted} for the test $\hat{\phi}_{\hat{\gamma}}$. Therefore
          \begin{equation}
   \frac{\partial  \lambda^* }{\partial \hat{\gamma}}=  \frac{\partial ^2	\hat{E}_1(\hat{\phi}_{\hat{\gamma}})  }{\partial \hat{\gamma}^2} \geq 0,
     \end{equation}
     where the inequality is from convexity of 	$\hat{E}_1(\hat{\phi}_{\hat{\gamma}})$ respect to $\hat{\gamma}$. Therefore, we only need to consider the behavior of variance as $\lambda$ changes. Taking the derivative of variance respect to $\lambda$, we have
		\begin{align}
		\frac{\partial }{\partial \lambda}\text{Var}_{\hat{P}_1} \Big(\frac{\hat{Q}_{\lambda}(X)}{\hat{P}_1(X)}  \Big)&=\sum_{x\in \mathcal{X}}  \frac{2{\hat{Q}_{\lambda}}(x)}{\hat{P_1}(x)} \frac{\partial \hat{Q}_{\lambda}(x) }{\partial \lambda}\\
		&= \sum_{x\in \mathcal{X}}  \frac{2{\hat{Q}_{\lambda}}(x)}{\hat{P_1}(x)}  \Bigg( \hat{Q}_{\lambda}(x)  \log \frac{\hat{P}_2(x)}{\hat{P}_1(x)}-\nonumber \\
		&~~~\hat{Q}_{\lambda}(x)  \sum_{x' \in \mathcal{X}} \hat{Q}_{\lambda}(x')\log \frac{\hat{P}_2(x')}{\hat{P}_1(x')} \Bigg ) \\
		&=2 \mathbb{E}_{\hat{Q}_{\lambda}} \bigg[ \frac{\hat{Q}_{\lambda}(X)}{\hat{P}_1(X)} \log  \frac{\hat{P}_2(X)}{\hat{P}_1(X)} \bigg  ]\nonumber \\
		&~~~- 2 \mathbb{E}_{\hat{Q}_{\lambda}} \bigg [ \frac{\hat{Q}_{\lambda}(X)}{\hat{P}_1(X)}\bigg ] \mathbb{E}_{\hat{Q}_{\lambda}} \bigg [ \log  \frac{\hat{P}_2(X)}{\hat{P}_1(X)} \bigg ].
		\end{align}
		Substituting $\hat{Q}_{\lambda}(X)$ as a function of $\lambda$ we get
		\begin{align}
		&\frac{\sum_{a\in \mathcal{X}} \hat{P}_1^{1-\lambda}(a) \hat{P}_2^{\lambda}(a)  }{2} \frac{\partial }{\partial \lambda}\text{Var}_{\hat{P}_1} \bigg(\frac{\hat{Q}_{\lambda}(X)}{\hat{P}_1(X)}  \bigg)\nonumber \\
		 &=\mathbb{E}_{\hat{Q}_{\lambda}} \bigg [ \bigg (\frac{\hat{P}_{2}(X)}{\hat{P}_1(X)} \bigg  )^{\lambda} \log  \frac{\hat{P}_2(X)}{\hat{P}_1(X)} \bigg  ]  \nonumber \\
		&~~~~~-  \mathbb{E}_{\hat{Q}_{\lambda}} \bigg [ \bigg (\frac{\hat{P}_2(X)}{\hat{P}_1(X)} \bigg ) ^{\lambda} \bigg ] \mathbb{E}_{\hat{Q}_{\lambda}} \bigg [ \log  \frac{\hat{P}_2(X)}{\hat{P}_1(X)}  \bigg ]. \label{eq:varlogsum}
		\end{align}
		Let $r(X)= \Big (\frac{\hat{P}_2(X)}{\hat{P}_1(X)}\Big )^{\lambda}$, then
		\begin{align}
		\mathbb{E}_{\hat{Q}_{\lambda}} &\bigg [ \bigg (\frac{\hat{P}_{2}(X)}{\hat{P}_1(X)} \bigg  )^{\lambda} \log  \frac{\hat{P}_2(X)}{\hat{P}_1(X)} \bigg  ]\nonumber \\
		& ~~~~~~~~~~~~~~~~-\mathbb{E}_{\hat{Q}_{\lambda}} \bigg [ \bigg (\frac{\hat{P}_2(X)}{\hat{P}_1(X)} \bigg ) ^{\lambda} \bigg ] \mathbb{E}_{\hat{Q}_{\lambda}} \bigg [ \log  \frac{\hat{P}_2(X)}{\hat{P}_1(X)}  \bigg ] \nonumber \\
		&=\frac{1}{\lambda} \mathbb{E}_{\hat{Q}_{\lambda}} [ r(X)  \log r(X) ]- \frac{1}{\lambda}  \mathbb{E}_{\hat{Q}_{\lambda}} [ r(X) ] \mathbb{E}_{\hat{Q}_{\lambda}} [ \log r(X)  ]. \label{eq:varlogsum}
		\end{align} 
		Note that $\hat{Q}_{\lambda}(x),r(x)$ are positive for all $x\in \Xc$. Therefore, using the log-sum inequality \cite{Cover} for the first term and Jensen inequality \cite{Cover} for the second term in \eqref{eq:varlogsum}, we obtain
		\begin{align}
		&\frac{\lambda \sum_{a\in \mathcal{X}} \hat{P}_1^{1-\lambda}(a) \hat{P}_2^{\lambda}(a)  }{2}\frac{\partial }{\partial \lambda}\text{Var}_{\hat{P}_1} \Big(\frac{\hat{Q}_{\lambda}(X)}{\hat{P}_1(X)}  \Big) \nonumber\\
	&\geq	\mathbb{E}_{\hat{Q}_{\lambda}} [ r(X) ] \log \mathbb{E}_{\hat{Q}_{\lambda}} [  r(X)  ]-  \mathbb{E}_{\hat{Q}_{\lambda}} [ r(X) ] \mathbb{E}_{\hat{Q}_{\lambda}} [ \log r(X)  ]  \\
	&\geq		\mathbb{E}_{\hat{Q}_{\lambda}} [ r(X) ] \log \mathbb{E}_{\hat{Q}_{\lambda}} [  r(X)  ]-  \mathbb{E}_{\hat{Q}_{\lambda}} [ r(X) ] \log \mathbb{E}_{\hat{Q}_{\lambda}} [  r(X)  ]\\
	&=0.
		\end{align} 
		Also, the above inequalities are met with equality when both log-sum and Jensen's inequalities are met with equality, which happens when $\lambda=0$. Therefore, for $ \lambda>0$,  $\text{Var}_{\hat{P}_1} \Big(\frac{\hat{Q}_{\lambda}(X)}{\hat{P}_1(X)}  \Big)$ is an increasing function of $\lambda$ for $ \lambda>0$ and consequently 
		\begin{equation}
			\frac{\partial }{\partial \hat{\gamma}}S_1(\hat{P}_1,\hat{P}_2,\hat{\gamma}) \geq 0.
		\end{equation}

\bibliographystyle{ieeebib}
\bibliographystyle{ieeetr}
\bibliography{journal_abbr,izs-2020}

\end{document}